\documentclass[conference,draftcls,onecolumn]{IEEEtran}

\usepackage{amsmath,amsfonts,amssymb,bbm,amsthm} 
\usepackage{multirow}
\usepackage{extarrows}
\usepackage{algorithm}
\usepackage{algorithmic}
\usepackage{enumerate}
\usepackage{threeparttable}
\usepackage{color}
\ifCLASSINFOpdf
\usepackage[pdftex]{graphicx}
\else
\fi
\usepackage{float}
\usepackage{cite}
\usepackage{epsfig}
\usepackage{mathrsfs}
\usepackage{array}
\usepackage{subfig}
\usepackage{verbatim}
\usepackage{setspace}
\usepackage{nicefrac}
\usepackage{arydshln}
\usepackage{changepage}
\usepackage[dvipsnames]{xcolor}
\usepackage{url}
\usepackage{fancyhdr}                                
\usepackage{lastpage}                                           
\usepackage{layout}

\DeclareGraphicsExtensions{.eps,.pdf,.jpg,.png}

\newtheorem{theorem}{Theorem}
\newtheorem{lemma}{Lemma}

\newtheorem{definition}{Definition}
\newtheorem{remark}{Remark}

\allowdisplaybreaks


\begin{document}

\title{Coded Sparse Matrix Multiplication}

\author{\IEEEauthorblockN{Sinong Wang$^{*}$, Jiashang Liu$^{*}$ and Ness Shroff$^{*\dag}$}
	\IEEEauthorblockA{$^{*}$Department of Electrical and Computer Engineering\\
		$^{\dag}$Department of Computer Science and Engineering\\
		The Ohio State University\\
		\{wang.7691, liu.3992 shroff.11\}@osu.edu}}

\maketitle

\begin{abstract}
	In a large-scale and distributed matrix multiplication problem $C=A^{\intercal}B$, where $C\in\mathbb{R}^{r\times t}$,  the coded computation plays an important role to effectively deal with ``stragglers'' (distributed computations that may get delayed due to few slow or faulty processors). However, existing coded schemes could destroy the significant sparsity that exists in large-scale machine learning problems, and could result in much higher computation overhead, i.e., $O(rt)$ decoding time. In this paper, we develop a new coded computation strategy, we call \emph{sparse code}, which achieves near \emph{optimal recovery threshold}, \emph{low computation overhead}, and \emph{linear decoding time} $O(nnz(C))$.  We implement our scheme and demonstrate the advantage of the approach over both uncoded and current fastest coded strategies.
	
\end{abstract}

\section{Introduction}

In this paper, we consider a distributed matrix multiplication problem, where we aim to compute $C=A^{\intercal}B$ from input matrices $A\in\mathbb{R}^{s\times r}$ and $B\in\mathbb{R}^{s\times t}$ for some integers $r,s,t$. This problem is the key building block in machine learning and signal processing problems, and has been used in a large variety of application areas including classification, regression, clustering and feature selection problems. Many such applications have large-scale datasets and massive computation tasks, which forces practitioners to adopt distributed computing frameworks such as Hadoop~\cite{dean2008mapreduce} and Spark~\cite{zaharia2010spark} to increase the learning speed. 

Classical approaches of distributed matrix multiplication rely on dividing the input matrices equally among all available worker nodes. Each worker computes a partial result, and the master node has to collect the results from all workers to output matrix $C$. As a result, a major performance bottleneck is the latency in waiting for a few slow or faulty processors -- called ``stragglers'' to finish their tasks~\cite{dean2013tail}. To alleviate this problem, current frameworks such as Hadoop deploy various straggler detection techniques and usually \emph{replicate} the straggling task on another available node.

Recently, forward error correction or coding techniques provide a more effective way to deal with the ``straggler'' in the distributed tasks~\cite{dutta2016short,lee2017speeding,tandon2017gradient,yu2017polynomial,li2018fundamental}. It creates and exploits coding redundancy in local computation to enable the matrix $C$ recoverable from the results of partial finished workers, and can therefore alleviate some straggling workers. For example,  
consider a distributed system with $3$ worker nodes, the coding scheme first splits the matrix $A$ into two submatrices, i.e., $A=[A_1,A_2]$. Then each worker computes $A_1^{\intercal}B$, $A_2^{\intercal}B$ and $(A_1+A_2)^{\intercal}B$.  The master node can compute $A^{\intercal}B$ as soon as \textbf{any} $2$ out of the $3$ workers finish, and can therefore overcome one straggler. 

In a general setting with $N$ workers, each input matrix $A$, $B$ is divided into $m$, $n$ submatrices, respectively.  The \emph{recovery threshold} is defined as the minimum number of workers that the master needs to wait for in order to compute $C$. The above MDS coded scheme is shown to achieve a recovery threshold $\Theta(N)$. An improved scheme proposed in~\cite{lee2017high} referred to as the product code, can offer a recovery threshold of $\Theta(mn)$ with high probability.  More recently, the work~\cite{yu2017polynomial} designs a type of \emph{polynomial code}. It achieves the recovery threshold of $mn$, which exactly matches the information theoretical lower bound. However, many problems in machine learning exhibit both extremely \textbf{large-scale} targeting data and a \textbf{sparse} structure, i.e., $nnz(A)\ll rs$, $nnz(B)\ll st$ and $nnz(C)\ll rt$. The key question that arises in this scenario is: \emph{is coding really an efficient way to mitigate the straggler in the distributed sparse matrix multiplication problem?}

\subsection{Motivation: Coding  Straggler}

To answer the aforementioned question, we first briefly introduce the current coded matrix multiplication schemes. In these schemes, each local worker calculates a coded version of submatrix multiplication. For example, $k$th worker of the polynomial code~\cite{yu2017polynomial} essentially calculates
\begin{equation}\label{eq:polycompute}
\underbrace{\sum\nolimits_{i=1}^{m} A_i^{\intercal} x_k^{i}}_{\tilde{A}_k^{\intercal}}\underbrace{\sum\nolimits_{j=1}^{n} B_j x_k^{jm}}_{\tilde{B}_k},
\end{equation}
where $A_i, B_j$ are the corresponding submatrics of the input matrices $A$ and $B$, respectively, and $x_k$ is a given integer. One can observe that, if $A$ and $B$ are sparse, due to the matrices additions, the density of the coded matrices $\tilde{A}_k$ and $\tilde{B}_k$ will increase at most $m$ and $n$ times, respectively. Therefore, the time of matrix multiplication $\tilde{A}_k^{\intercal}\tilde{B}_k$ will increase roughly $O(mn)$ times of the simple uncoded one $A_i^{\intercal}B_j$.

\begin{figure}[t]
	\vskip -0.05in
	\begin{center}
		\centerline{\includegraphics[width=5in]{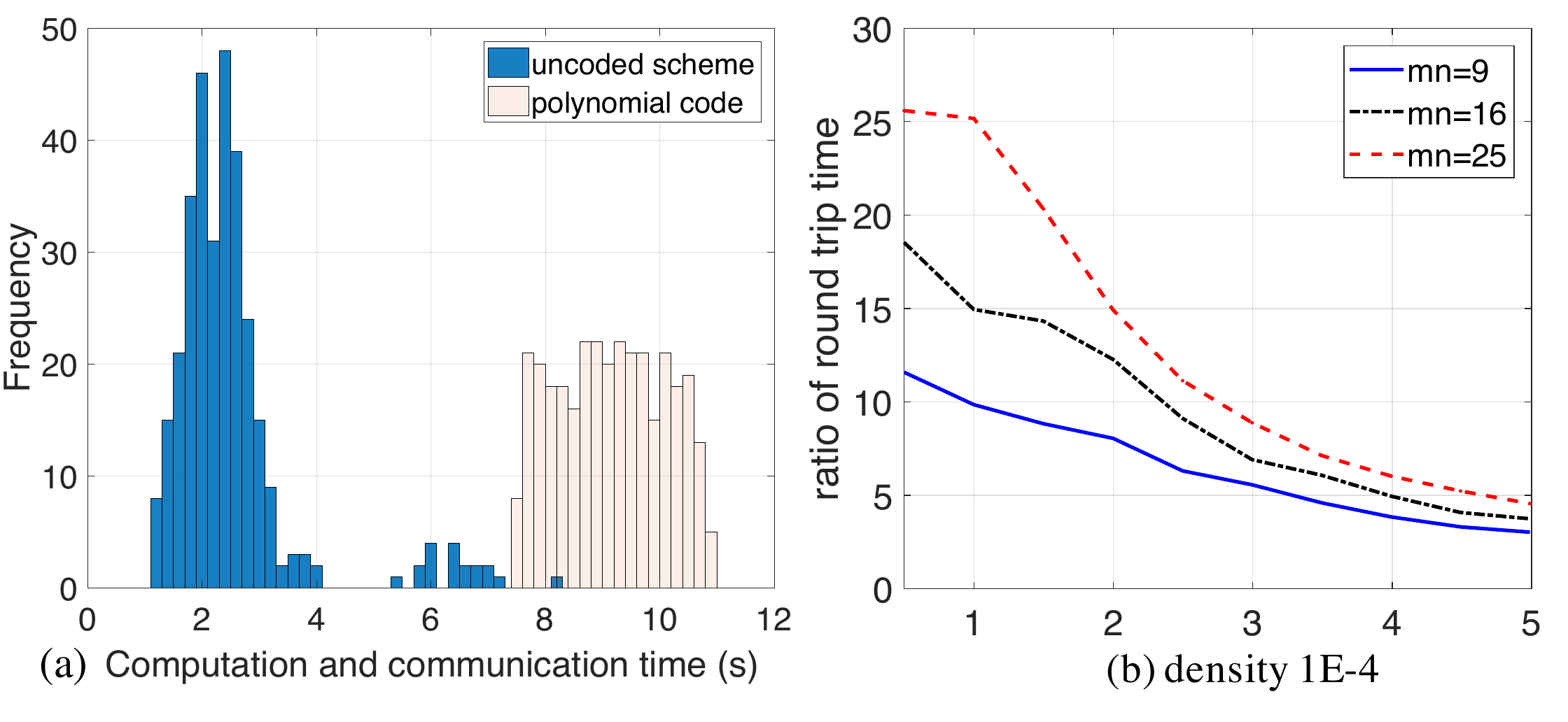}}
		\vskip -0.1in
		\caption{Measured local computation and communication time.}
		\vskip -0.5in
		\label{fig:motivation}
	\end{center}
\end{figure}

In the Figure~\ref{fig:motivation}(a), we experiment large and sparse matrix multiplication from two random Bernoulli square matrices with dimension roughly equal to $1.5\times10^5$ and the number of nonzero elements equal to $6\times10^5$. We show measurements on local computation and communication time required for $N=16$ workers to operate the polynomial code and uncoded scheme. Our finding is that the final job completion time of the polynomial code is significantly increased compared to that of the uncoded scheme. The main reason is that the increased density of input matrix leads to the increased computation time, which further incurs the even larger data transmission and higher I/O contention. In the Figure~\ref{fig:motivation}(b), we generate two $10^5$ random square Bernoulli matrices with different densities $p$. We plot the ratio of the average local computation time between the polynomial code and the uncoded scheme versus the matrix density $p$. It can be observed that the coded scheme requires $5$ or more computation time vs the uncoded scheme, and this ratio is particularly large, i.e., $O(mn)$, when the matrix is sparse. 

In certain suboptimal coded computation schemes such as MDS code, product code~\cite{lee2017speeding,lee2017high} and short dot~\cite{dutta2016short}, the generator matrices are generally dense, which also implies a heavy workload for each local worker. Moreover, the decoding algorithm of polynomial code and MDS type of codes is based on the fast polynomial interpolation algorithm in the finite field. Although it leads to nearly linear decoding time, i.e., $O(rt\ln^2(mn\ln(mn)))$, it still incurs extremely high cost when dealing with current large-scale and sparse data.  Therefore, inspired by this phenomenon, we are interested in the following key problem: \emph{can we find a coded matrix multiplication scheme that has small recovery threshold, low computation overhead and decoding complexity only dependent on $nnz(C)$?}

\subsection{Main Contribution}

In this paper, we answer this question positively by designing a novel coded computation strategy, we call \emph{sparse code}. It achieves near optimal recovery threshold $\Theta(mn)$ by exploiting the coding advantage in local computation. Moreover, such a coding scheme can exploit the sparsity of both input and output matrices, which leads to low computation load, i.e.,  $O(\ln(mn))$ times of uncoded scheme and, nearly linear decoding time $O(nnz(C)\ln(mn))$. 


The basic idea in \emph{sparse code} is: each worker chooses a random number of input submatrices based on a given degree distribution $P$; then computes a weighted linear combination $\sum_{ij} w_{ij}A_i^{\intercal}B_j$, where the weights $w_{ij}$ are randomly drawn from a finite set $S$. When the master node receives a bunch of finished tasks such that the coefficient matrix formed by weights $w_{ij}$ is full rank, it starts to operate a hybrid decoding algorithm between peeling decoding and Gaussian elimination to recover the resultant matrix $C$. 

We prove the optimality of the \emph{sparse code} by carefully designing the degree distribution $P$ and the algebraic structure of set $S$. The recovery threshold of the sparse code is mainly determined by how many tasks are required such that the coefficient matrix is full rank and the hybrid decoding algorithm recovers all the results. We design a type of Wave Soliton distribution (definition is given in Section~\ref{sec:theory}), and show that, under such a distribution, when $\Theta(mn)$ tasks are finished, the hybrid decoding algorithm will successfully decode all the results with decoding time $O(nnz(C)\ln(mn))$.

Moreover, we reduce the full rank analysis of the coefficient matrix to the determinant analysis of a random matrix in $\mathbb{R}^{mn\times mn}$. The state-of-the-art in this field is limited to the Bernoulli case~\cite{tao2007singularity, bourgain2010singularity}, in which each element is identically and independently distributed random variable. However, in our proposed sparse code, the matrix is generated from a degree distribution, which leads to dependencies among the elements in the same row. To overcome this difficulty, we find a different technical path: we first utilize the Schwartz-Zeppel Lemma~\cite{schwartz1980fast} to reduce the determinant analysis problem to the analysis of the probability that a random bipartite graph contains a perfect matching.  Then we combine the combinatoric graph theory and the probabilistic method to show that when number of $mn$ tasks are collected, the coefficient matrix is full rank with high probability.


We further utilize the above analysis to formulate an optimization problem to determine the optimal degree distribution $P$ when $mn$ is small. We finally implement and benchmark the sparse code on Ohio Super Computing Center~\cite{OhioSupercomputerCenter1987}, and empirically demonstrate its performance gain compared with the existing strategies.

\section{Preliminary}

We are interested in a matrix multiplication problem with two input matrices $A\in\mathbb{R}^{s\times r}$, $B\in\mathbb{R}^{s\times t}$ for some integers $r,s,t$. Each input matrix $A$ and $B$ is evenly divided along the column side into $m$ and $n$ submatrices, respectively.
\begin{equation}\label{eq:division}
A=[A_1,A_2,\ldots,A_m]\text{ and } B=[B_1,B_2,\ldots,B_n].
\end{equation}
Then computing the matrix $C$ is equivalent to computing $mn$ blocks $C_{ij}=A_i^{\intercal}B_j$. Let the set $W=\{C_{ij}=A_i^{\intercal}B_j|1\leq i \leq m, 1\leq j \leq n\}$ denote these components.
Given this notation, the \emph{coded distributed matrix multiplication problem} can be described as follows: define $N$ coded computation functions, denoted by
\begin{equation*}
\boldsymbol{f}=(f_1,f_2,\ldots,f_N).
\end{equation*}
Each local function $f_i$ is used by worker $i$ to compute a submatrix $\tilde{C}_{i}\in\mathbb{R}^{\frac{r}{m}\times \frac{t}{n}}=f_i(W)$ and return it to the master node. The master node waits only for the results of the partial workers $\{\tilde{C}_{i}|i\in I\subseteq\{1,\ldots,N\}\}$ to recover the final output $C$ using certain decoding functions. For any integer $k$, the recovery threshold $k(\boldsymbol{f})$ of a coded computation strategy $\boldsymbol{f}$ is defined as the minimum integer $k$ such that the master node can recover matrix $C$ from results of the any $k$ workers. The framework is illustrated in Figure~~\ref{fig:model}.

\begin{figure}[htb]
	\begin{center}
		\centerline{\includegraphics[width=4.2in]{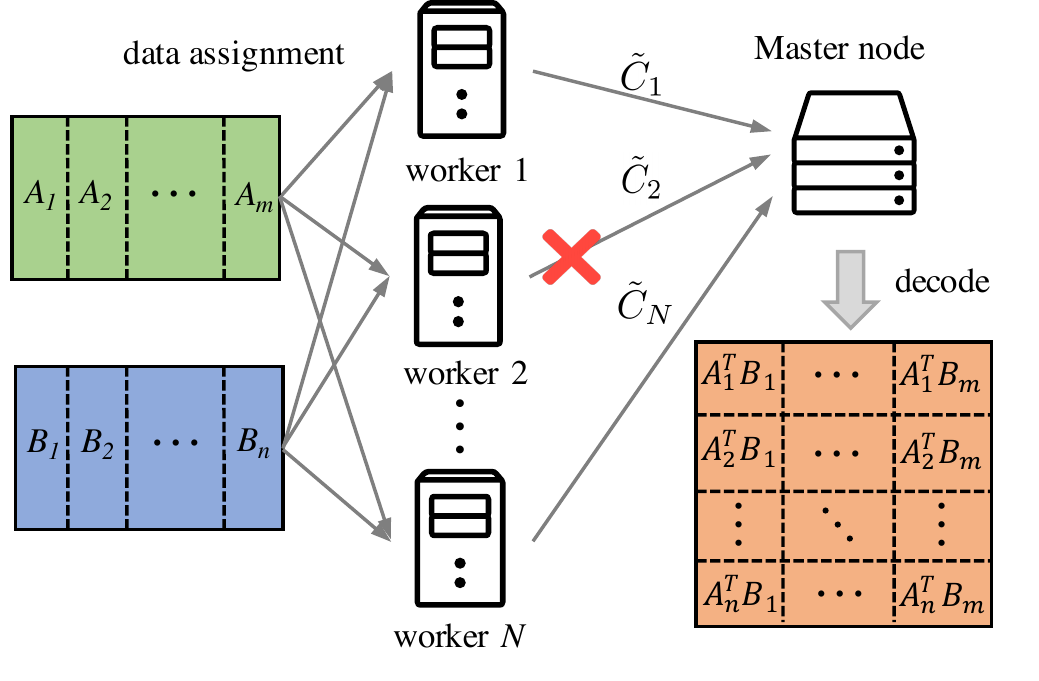}}
		\vskip -0.2in
		\caption{Framework of coded distributed matrix multiplication.}
		\label{fig:model}
	\end{center}
	\vskip -0.4in
\end{figure}

\subsection{Main Results}

The main result of this paper is the design of a new coded computation scheme, we call the \emph{sparse code}, that has the following performance.
\begin{theorem}\label{thm:main}
	The sparse code achieves a recovery threshold $\Theta(mn)$ with high probability, while allowing nearly linear decoding time $O(nnz(C)\ln(mn))$ at the master node.
\end{theorem}
As shown in TABLE~\ref{tab:comparison}, compared to the state of the art, the sparse code provides order-wise improvement in terms of the recovery threshold, computation overhead and decoding complexity. Specifically, the decoding time of MDS code~\cite{lee2017speeding}, product code~\cite{lee2017high}, LDPC code and polynomial code~\cite{yu2017polynomial} is $O(rt\ln^2(mn\ln(mn)))$, which is dependent on the dimension of the output matrix. Instead, the proposed sparse code actually exhibits a decoding complexity that is nearly linear time in number of nonzero elements of the output matrix $C$, which is extremely less than the product of dimension. Although the decoding complexity of sparse MDS code is linear to $nnz(C)$, it is also dependent on the $mn$. To the best of our knowledge, this is the first coded distributed matrix multiplication scheme with complexity independent of the dimension.

\begin{table}[t]
	\small
	\vskip -0.15in
	\caption{Comparison of Existing Coding Schemes}
	\vskip 0.05in
	\label{tab:comparison}
	\centering
	\begin{threeparttable}
		\begin{tabular}{|c|c|c|c|}
			\hline
			Scheme & Recovery Threshold & Computation Overhead & Decoding Complexity\\
			\hline
			MDS code~\cite{lee2017speeding} &  $\Theta(N)$ & $\Theta(mn)$  & $\tilde{O}(rt)\tnote{2}$  \\
			\hline
			sparse MDS code~\cite{lee2017coded} &  $\Theta^*(mn)$\tnote{2} & $\Theta(\ln(mn))$  & $\tilde{O}(mn\cdot nnz(C))$  \\
			\hline
			product code~\cite{lee2017high} & $\Theta^*(mn)$ & $\Theta(mn)$ & $\tilde{O}(rt)$   \\
			\hline
			LDPC code~\cite{lee2017high} & $\Theta^*(mn)$ &$\Theta(\ln(mn))$ & $\tilde{O}(rt)$   \\
			\hline
			polynomial code~\cite{yu2017polynomial} &  $mn$ & $mn$  & $\tilde{O}(rt)$  \\
			\hline
			our scheme &  $\Theta^*(mn)$ & $\Theta(\ln(mn))$  & $\tilde{O}(nnz(C))$ \\
			\hline
		\end{tabular}
		\begin{tablenotes}
			\scriptsize
			\item[1] Computation overhead is the time of local computation over uncoded scheme.
			\item[2] $\tilde{O}(\cdot)$ omits the logarithmic terms and $O^*(\cdot)$ refers the high probability result.
		\end{tablenotes}
	\end{threeparttable}
	\vskip -0.2in
\end{table}

Regarding the recovery threshold, the existing work~\cite{yu2017polynomial} has applied a cut-set type argument to show that the minimum recovery threshold of any scheme is
\begin{equation}
K^*=\min\limits_{\boldsymbol{f}} k(\boldsymbol{f}) =  mn.
\end{equation}
The proposed sparse code matches this lower bound with a constant gap and high probability. 

\section{Sparse Codes}

In this section, we first demonstrate the main idea of the sparse code through a motivating example. We then formally describe the construction of the general sparse code and its decoding algorithm. 

\subsection{Motivating Example}

Consider a distributed matrix multiplication task $C=A^{\intercal}B$ using $N=6$ workers. Let $m=2$ and $n=2$ and each input matrix $A$ and $B$ be evenly divided as
\begin{equation*}
A=[A_1, A_2]\quad\text{ and } \quad B=[B_1, B_2].
\end{equation*}
Then computing the matrix $C$ is equivalent to computing following $4$ blocks.
\begin{equation*}
C=A^{\intercal}B=\begin{bmatrix}
A_1^{\intercal}B_1 & A_1^{\intercal}B_2\\
A_2^{\intercal}B_1 & A_2^{\intercal}B_2
\end{bmatrix}
\end{equation*}

We design a coded computation strategy via the following procedure: each worker $i$ locally computes a weighted sum of four components in matrix $C$.
\begin{equation*}
\tilde{C}_i=w^i_1 A_1^{\intercal}B_1 +w^i_2 A_1^{\intercal}B_2 + w^i_3 A_2^{\intercal}B_1 + w^i_4 A_2^{\intercal}B_2,
\end{equation*}
Each weight $w^i_j$ is independently and identically distributed Bernoulli random variable with parameter $p$. For example, Let $p=1/3$, then on average, $2/3$ of these weights are equal to $0$. We randomly generate the following $N=6$ local computation tasks.
\begin{align*}
&\tilde{C}_1 = A_1^{\intercal}B_1 + A_1^{\intercal}B_2, \quad \tilde{C}_2 = A_1^{\intercal}B_2 + A_2^{\intercal}B_1\\
&\tilde{C}_3 = A_1^{\intercal}B_1,\quad \tilde{C}_4 = A_1^{\intercal}B_2 + A_2^{\intercal}B_2\\
&\tilde{C}_5 = A_2^{\intercal}B_1 + A_2^{\intercal}B_2, \quad \tilde{C}_6 = A_1^{\intercal}B_1 + A_2^{\intercal}B_1
\end{align*}

\begin{figure*}[t]
	\begin{center}
		\centerline{\includegraphics[width=6.5in]{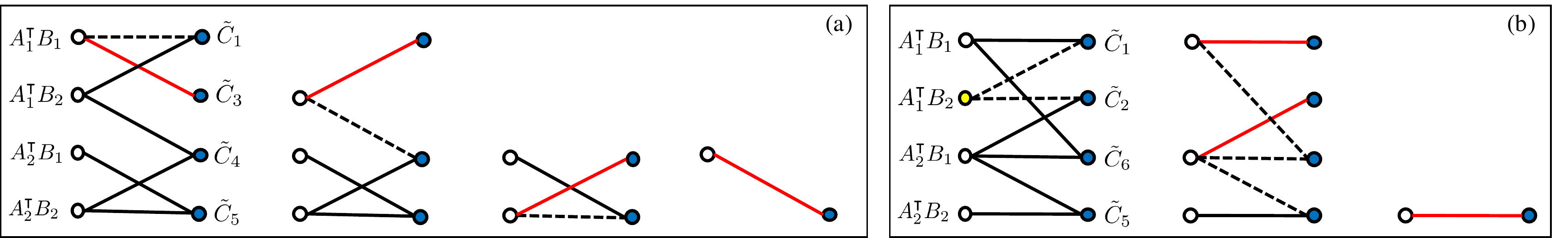}}
		\vskip -0.1in
		\caption{Example of the hybrid peeling and Gaussian decoding process of sparse code.}
		\label{fig:example}
	\end{center}
	\vskip -0.4in
\end{figure*}

Suppose that both the $2$rd and $6$th workers are stragglers and the master node has collected the results from nodes $\{1,3,4,5\}$. According to the designed computation strategy, we have following group of linear systems.
\begin{equation*}
\begin{bmatrix}
\tilde{C}_1\\
\tilde{C}_3\\
\tilde{C}_4\\
\tilde{C}_5
\end{bmatrix} = \begin{bmatrix}
1 & 1 & 0 & 0\\
1 & 0 & 0 & 0\\
0 & 1 & 0 & 1\\
0 & 0 & 1 & 1
\end{bmatrix}\cdot\begin{bmatrix}
A_1^{\intercal}B_1\\
A_1^{\intercal}B_2\\
A_2^{\intercal}B_1\\
A_2^{\intercal}B_2
\end{bmatrix}
\end{equation*}
One can easily check that the above coefficient matrix is full rank. Therefore, one straightforward way to recover $C$ is to solve $rt/4$ linear systems, which proves decodability. However, the complexity of this decoding algorithm is expensive, i.e., $O(rt)$ in this case. 

Interestingly, we can use a type of \emph{peeling algorithm} to recover the matrix $C$ with only three sparse matrix additions: first, we can straightforwardly recover the block $A_1^{\intercal}B_1$ from worker $3$. Then we can use the result of worker $1$ to recover block $A_1^{\intercal}B_2 = \tilde{C}_1-A_1^{\intercal}B_1$. Further, we can use the results of worker $4$ to recover block $A_2^{\intercal}B_2=\tilde{C}_4-A_1^{\intercal}B_2$ and use the results of worker $5$ to obtain block $A_2^{\intercal}B_1=\tilde{C}_5-A_2^{\intercal}B_2$. Actually, the above peeling decoding algorithm can be viewed as an edge-removal process in a bipartite graph.  We construct a bipartite graph with one partition being the original blocks $W$ and the other partition being the finished coded computation tasks $\{\tilde{C}_i\}$. Two nodes are connected if such computation task contains that block. As shown in the Figure~\ref{fig:example}(a), in each iteration, we find a \emph{ripple} (degree one node) in the right that can be used to recover one node of left. We remove the adjacent edges of that left node, which might produce some new ripples in the right. Then we iterate this process until we decode all blocks. 

Based on the above graphical illustration, the key point of successful decoding is the existence of the ripple during the edge removal process. Clearly, this is not always true from the design of our coding scheme and the uncertainty in the cloud. For example, if both the $3$rd and $4$th workers are stragglers and the master node has collected the results from node $\{1,2,5,6\}$, even though the coefficient matrix is full rank, there exists no ripple in the graph. To avoid this problem, we can randomly pick one block and recover it through a linear combination of the collected results, then use this block to continue the decoding process. This particular linear combination can be determined by solving a linear system. Suppose that we choose to recover $A_1^{\intercal}B_2$,  then we can recover it via the following linear combination.
\begin{equation*}
A_1^{\intercal}B_2=\frac{1}{2}\tilde{C}_1+\frac{1}{2}\tilde{C}_2-\frac{1}{2}\tilde{C}_6.
\end{equation*}
As illustrated in the Figure~\ref{fig:example}(b), we can recover the rest of the blocks using the same peeling decoding process. The above decoding algorithm only involves simple matrix additions and the total decoding time is $O(nnz(C))$. 

\subsection{General Sparse Code}

Now we present the construction and decoding of the sparse code in a general setting. We first evenly divide the input matrices $A$ and $B$ along the column side into $m$ and $n$ submatrices, as defined in (\ref{eq:division}). Then we define a set $S$ that contains $m^2n^2$ distinct elements except zero element. One simplest example of  $S$ is $[m^2n^2]\triangleq\{1,2,\ldots,m^2n^2\}$. Under this setting, we define the following class of coded computation strategies.
\begin{definition}\label{def:spcode}\emph{(Sparse Code)}
	Given the parameter $P\in \mathbb{R}^{mn}$ and set $S$, we define the $(P,S)-$sparse code as: for each worker $k\in[N]$, compute
	\begin{equation}\label{eq:linearcombine}
	\tilde{C}_k=f_k(W) = \sum\limits_{i=1}^{m}\sum\limits_{j=1}^{n} w^k_{ij}A_i^{\intercal}B_j.
	\end{equation}
	Here the parameter $P=[p_1,p_2,\ldots,p_{mn}]$ is the degree distribution, where the $p_l$ is the probability that there exists number of $l$ nonzero weights $w^k_{ij}$ in each worker $k$. The value of each nonzero weight $w^{k}_{ij}$ is picked from set $S$ independently and uniformly at random.
\end{definition}

Without loss of generality, suppose that the master node collects results from the first $K$ workers with $K\leq N$. Given the above coding scheme, we have
\begin{equation*}
\begin{bmatrix}
\tilde{C}_1\\
\tilde{C}_2\\
\vdots\\
\tilde{C}_K
\end{bmatrix} = \begin{bmatrix}
w^1_{11} & w^1_{12} & \cdots & w^1_{mn}\\
w^2_{11} & w^2_{12} & \cdots & w^2_{mn}\\
\vdots & \vdots & \ddots & \vdots\\
w^K_{11} & w^K_{12} & \cdots & w^K_{mn}
\end{bmatrix}\cdot\begin{bmatrix}
A_1^{\intercal}B_1\\
A_1^{\intercal}B_2\\
\vdots\\
A_m^{\intercal}B_n
\end{bmatrix}.
\end{equation*}
We use $M\in \mathbb{R}^{K\times mn}$ to represent the above coefficient matrix. To guarantee decodability, the master node should collect results from enough number of workers such that the coefficient matrix $M$ is of column full rank. Then the master node goes through a \textbf{peeling decoding} process: it first finds a ripple worker to recover one block. Then for each collected results, it subtracts this block if the computation task contains this block. If there exists no ripple in our peeling decoding process, we go to \textbf{rooting step}: randomly pick a particular block $A_i^{\intercal}B_j$. The following lemma shows that we can recover this block via a linear combination of the results $\{\tilde{C}_k\}_{k=1}^{K}$.
\begin{lemma} \emph{(rooting step)}
	If rank$(M)=mn$, for any $k_0\in\{1,2,\ldots,mn\}$, we can recover a particular block $A_i^{\intercal}B_j$ with column index $k_0$ in matrix $M$ via the following linear combination.
	\begin{equation}\label{eq:borrow}
	A_i^{\intercal}B_j=\sum\nolimits_{k=1}^K u_{k} \tilde{C}_k.
	\end{equation}
	The vector $u=[u_1,\ldots,u_K]$ can be determined by solving $M^Tu=e_{k_0}$, where $e_{k_0}\in\mathbb{R}^{K}$ is a unit vector with unique $1$ locating at the index $k_0$. 
\end{lemma}

The basic intuition is to find a linear combination of row vectors of matrix $M$ such that the row vectors eliminate all other blocks except the particular block $A_i^{\intercal}B_j$. The whole procedure is listed in Algorithm~\ref{alg:spcode}.

Here we conduct some analysis of the complexity of Algorithm~\ref{alg:spcode}. During each iteration, the complexity of operation $\tilde{C}_{k}=\tilde{C}_{k}-M_{kk_0}A_i^{\intercal}B_j$ is $O(nnz(A_i^{\intercal}B_j))$. Suppose that the number of average nonzero elements in each row of coefficient matrix $M$ is $\alpha$. Then each block $A_i^{\intercal}B_j$ will be used $O(\alpha K/mn)$ times in average. Further, suppose that there exists number of $c$ blocks requiring the rooting step (\ref{eq:borrow}) to recover, the complexity in each step is $O(\sum_{k}nnz(\tilde{C}_k))$. On average, each coding block $\tilde{C}_k$ is equal to the sum of $O(\alpha)$ original blocks. Therefore, the complexity of Algorithm~\ref{alg:spcode} is
\begin{align}
&O\left(\frac{\alpha K}{mn}\sum\nolimits_{i,j}nnz(A_i^{\intercal}B_j)\right)+O\left(c\sum\nolimits_{k}nnz(\tilde{C}_k)\right)\notag\\
=&O\left((c+1)\alpha K/mn\cdot nnz(C)\right).\label{eq:timecomplexity}
\end{align}
We can observe that the decoding time is linear in the density of matrix $M$, the recovery threshold $K$ and the number of rooting steps (\ref{eq:borrow}). In the next section, we will show that, under a good choice of degree distribution $P$ and set $S$, we can achieve the result in Theorem~\ref{thm:main}.
\begin{algorithm}[tb]
	\caption{Sparse code (master node's protocol)}
	\label{alg:spcode}
	\begin{algorithmic}
		\REPEAT
		\STATE The master node assign the coded computation tasks according to Definition~\ref{def:spcode}.
		\UNTIL{the master node collects results with rank($M$)$=mn$ and $K$ is larger than a given threshold}.
		\REPEAT
		\STATE Find a row $M_{k'}$ in matrix $M$ with $\|M_{k'}\|_0=1$.
		\IF{such row does not exist }
		\STATE Randomly pick a $k_0\in \{1,\ldots,mn\}$ and recover corresponding block $A_i^{\intercal}B_j$ by (\ref{eq:borrow}).
		\ELSE
		\STATE Recover the block $A_i^{\intercal}B_j$ from $\tilde{C}_{k'}$.
		\ENDIF
		\STATE Suppose that the column index of the recovered block $A_i^{\intercal}B_j$ in matrix $M$ is $k_0$.
		\FOR{each computation results $\tilde{C}_{k}$}
		\IF{$M_{kk_0}$ is nonzero}
		\STATE $\tilde{C}_{k}=\tilde{C}_{k}-M_{kk_0}A_i^{\intercal}B_j$ and set $M_{kk_0}=0$.
		\ENDIF
		\ENDFOR
		\UNTIL{every block of matrix $C$ is recovered.}
	\end{algorithmic}
\end{algorithm}

\section{Theoretical Analysis\label{sec:theory}}

As discussed in the preceding section, to reduce the decoding complexity, it is good to make the coefficient matrix $M$ as sparse as possible. However, the lower density will require that the master node collects a larger number of workers to enable the full rank of matrix $M$. For example, in the extreme case, if we randomly assign one nonzero element in each row of $M$. The analysis of the classical balls and bins process implies that, when $K=O(mn\ln(mn))$, the matrix $M$ is full rank, which is far from the optimal recovery threshold. On the other hand, the polynomial code~\cite{yu2017polynomial} achieves the optimal recovery threshold. Nonetheless, it exhibits the densest matrix $M$, i.e., $K\times mn$ nonzero elements, which significantly increases the local computation, communication  and final decoding time. 

In this section, we will design the sparse code between these two extremes. This code has near optimal recovery threshold $K=\Theta(mn)$ and constant number of rooting steps (\ref{eq:borrow}) with high probability and extremely sparse matrix $M$ with $\alpha=\Theta(\ln(mn))$ nonzero elements in each row. The main idea is to choose the following degree distribution. 

\begin{definition}\emph{(Wave Soliton distribution)}
	The Wave Soliton distribution $P_w=[p_1,p_2,\ldots,p_{mn}]$ is defined as follows.
	\begin{equation}\label{eq:wave}
	p_k=\left\{
	\begin{aligned}
	&\frac{\tau}{mn}, k=1; \frac{\tau}{70}, k=2\\
	&\frac{\tau}{k(k-1)}, 3\leq k\leq mn
	\end{aligned}
	\right..
	\end{equation}
	The parameter $\tau = 35/18$ is the normalizing factor.
\end{definition}
The above degree distribution is modified from the Soliton distribution~\cite{luby2002lt}. In particular, we cap the original Soliton distribution at the maximum degree $mn$, and remove a constant weight from degree $2$ to other larger degrees. It can be observed that the recovery threshold $K$ of the proposed sparse code depends on two factors: (i) the full rank of coefficient matrix $M$; (ii) the successful decoding of peeling algorithm with constant number of rooting steps. 

\subsection{Full Rank Probability\label{sec:fullrank}}

Our first main result is to show that when $K=mn$, the coefficient matrix $M$ is full rank with high probability. Suppose that $K=mn$, we can regard the formation of the matrix $M$ via the following random graph model.
\begin{definition} \emph{(Random balanced bipartite graph)} Let $G(V_1,V_2,P)$ be a random blanced bipartite graph, in which $|V_1|=|V_2|=mn$. Each node $v\in V_2$ independently and randomly connects to $l$ nodes in partition $V_1$ with probability $p_l$.
\end{definition}
Define an Edmonds matrix $M(x)\in\mathbb{R}^{mn\times mn}$ of graph $G(V_1,V_2,P)$ with $[M(x)]_{ij}=x_{ij}$ if vertices $v_i\in V_1$, $v_j\in V_2$ are connected, and $[M(x)]_{ij}=0,$ otherwise. The coefficient matrix $M$ can be obtained by assigning each $x_{ij}$ a value from $S$ independently and uniformly at random. Then the probability that matrix $M$ is full rank is equal to the probability that the determinant of the Edmonds matrix $M(x)$ is nonzero at the assigning values $x_{ij}$. The following technical lemma~\cite{schwartz1980fast} provides a simple lower bound of the such an event.
\begin{lemma} \emph{(Schwartz-Zeppel Lemma)} Let $f(x_1,\ldots,x_N)$ be a nonzero polynomial with degree $d$. Let $S$ be a finite set in $\mathbb{R}$ with $|S|=d^2$. If we assign each variable a value from $S$ independently and uniformly at random, then
	\begin{equation}
	\mathbb{P}(f(x_1,x_2,\ldots,x_N)\neq 0)\geq 1- d^{-1}.
	\end{equation}
\end{lemma}
A classic result in graph theory is that \emph{a balanced bipartite graph contains a perfect matching if and only if the determinant of Edmonds matrix, i.e., $|M(x)|$,  is a nonzero polynomial}. Combining this result with Schwartz-Zeppel Lemma, we can finally reduce the analysis of the full rank probability of the coefficient matrix $M$ to the probability that the random graph $G(V_1,V_2,P_w)$ contains a perfect matching.
\begin{align}
\mathbb{P}(|M|\neq 0)=&\underbrace{\mathbb{P}(|M|\neq 0\big||M(x)|\not\equiv​ 0)}_{\text{ S-Z Lemma: }\geq 1-1/mn}\cdot\underbrace{\mathbb{P}(|M(x)|\not\equiv​ 0)}_{\text{perfect matching}}\notag
\end{align} 
Formally, we have the following result about the existence of the perfect matching in the above defined random bipartite graph.
\begin{theorem}\label{thm:fullrank}\emph{(Existence of perfect matching)}
	If the graph $G(V_1,V_2, P)$ is generated under the Wave Soliton distribution (\ref{eq:wave}), then there exists a constant $c>0$ such that
	\begin{equation*}
	\mathbb{P}(\text{$G$ contains a perfect matching})> 1-c(mn)^{-0.94}.
	\end{equation*}
\end{theorem}
\begin{proof}
	Here we sketch the proof. Details can be seen in the supplementary material. The basic technique is to utilize Hall's theorem to show that such a probability is lower bounded by the probability that $G$ does not contain a structure $S\subset V_1$ or $S\subset V_2$ such that $|S|>|N(S)|$, where $N(S)$ is the neighboring set of $S$. We show that, if $S\subset V_1$, the probability that $S$ exists is upper bounded by 
	\begin{equation*}
	\sum\limits_{s=\Theta(1)}\frac{1}{(mn)^{0.94s}}+\sum\limits_{s=\Omega(1)}^{s=o(mn)} \left(\frac{s}{mn}\right)^{0.94s}+\sum\limits_{s=\Theta(mn)}c_1^{mn}.
	\end{equation*}
	If $S\subset V_2$, the probability that $S$ exists is upper bounded by 
	\begin{equation*}
	\sum\limits_{s=\Theta(1)}\frac{1}{mn}+\sum\limits_{s=o(mn)} \frac{sc_2^s}{mn}+\sum\limits_{s=\Theta(mn)}c_3^{mn}.
	\end{equation*}
	where the constants $c_1,c_2,c_3$ are strictly less than 1. Combining these results together, gives us Theorem~\ref{thm:fullrank}.
\end{proof}
Analyzing the existence of perfect matching in a random bipartite graph has been developed since~\cite{erdos1964random}. However, existing analysis is limited to the independent generation model. For example, the Erdos-Renyi model assumes each edge exists independently with probability $p$. The $\kappa-$out model~\cite{walkup1980matchings} assumes each vertex $v\in V_1$ independently and randomly chooses $\kappa$ neighbors in $V_2$. The minimum degree model~\cite{frieze2004perfect} assumes each vertex has a minimum degree and edges are uniformly distributed among all allowable classes.
There exists no work in analyzing such a probability in a random bipartite graph generated by a given degree distribution. In this case, one technical difficulty is that each node in the partition $V_1$ is dependent. All analysis should be carried from the nodes of the right partition, which exhibits an intrinsic complicated statistical model.

\subsection{Optimality of  Recovery Threshold\label{sec:recovery}}

We now focus on quantitatively analyzing the impact of the recovery threshold $K$ on the peeling decoding process and the number of rooting steps (\ref{eq:borrow}). Intuitively, a larger $K$ implies a larger number of ripples, which leads to higher successful peeling decoding probability and therefore less number of rooting steps. The key question is: how large must $K$ be such that all $mn$ blocks are recovered with only constant number of rooting steps. To answer this question, we first define a distribution generation function of $P_w$ as
\begin{equation}\label{eq:generationfunction}
\Omega_w(x)=\frac{\tau}{mn}x+\frac{\tau}{70}x^2+\tau\sum\limits_{k=3}^{mn}\frac{x^k}{k(k-1)}.
\end{equation}
The following technical lemma is useful in our analysis.

\begin{lemma}\label{lm:decode}
	If the degree distribution $\Omega_w(x)$ and recovery threshold $K$ satisfy
	\begin{equation}\label{eq:decode}
	\left[1-\Omega_w'(1-x)/mn\right]^{K-1}\leq x, \text{ for } x\in[b/mn,1],
	\end{equation}
	then the peeling decoding process in Algorithm~\ref{alg:spcode} can recover $mn-b$ blocks with probability at least $1-e^{-cmn}$, where $b, c$ are constants.
\end{lemma}
Lemma~\ref{lm:decode} is tailored from applying a martingale argument to the peeling decoding process~\cite{luby2001efficient}.  This result provides a quantitative recovery condition on the degree generation function. It remains to be shown that the proposed Wave Soliton distribution (\ref{eq:generationfunction}) satisfies the above inequality with a specific choice of $K$.
\begin{theorem}\label{thm:decode}\emph{(Recovery threshold)} Given the sparse code with parameter$(P_w,[m^2n^2])$, if $K=\Theta(mn)$, then there exists a constant $c$ such that with probability at least $1-e^{-cmn}$, Algorithm~\ref{alg:spcode} is sufficient to recover all $mn$ blocks with $\Theta(1)$ blocks recovering from rooting step (\ref{eq:borrow}).
\end{theorem}
Combining the results of Theorem~\ref{thm:fullrank} and Theorem~\ref{thm:decode}, we conclude that the recovery threshold of sparse code $(P_w,[m^2n^2])$ is $\Theta(mn)$ with high probability. Moreover, since the average degree of Wave Soliton Distribution is $O(\ln(mn))$, combining these results with (\ref{eq:timecomplexity}), the complexity of Algorithm~\ref{alg:spcode} is therefore $O(nnz(C)\ln(mn))$. 
\begin{remark}
	Although the recovery threshold of the proposed scheme exhibits a constant gap to the information theoretical lower bound, the practical performance is very close to such a bound. This mainly comes from the pessimistic estimation in Theorem~\ref{thm:decode}. As illustrated in Figure~\ref{fig:theorypractice}, we generate the sparse code under Robust Soliton distribution, and plot the average recovery threshold versus the number of blocks $mn$. It can be observed that the overhead of proposed sparse code is less than $15\%$.
\end{remark}
\begin{remark}
	Existing codes such as Tornado code~\cite{luby1998tornado} and LT code~\cite{luby2002lt} also utilize the peeling decoding algorithm and can provide a recovery threshold $\Theta(mn)$. However, they exhibit a large constant, especially when $mn$ is less than $10^3$. Figure~\ref{fig:theorypractice} compares the practical recovery threshold among these codes. We can see that our proposed sparse code results in a much lower recovery threshold. Moreover, the intrinsic cascading structure of these codes will also destroy the sparsity of input matrices.
\end{remark}

\subsection{Optimal Design of Sparse Code }

The proposed Wave Soliton distribution (\ref{eq:wave}) is asymptotically optimal, however, it is far from optimal in practice when $m,n$ is small. The analysis of the full rank probability and the decoding process relies on an asymptotic argument to enable upper bounds of error probability. Such bounds are far from tight when $m, n$ are small. In this subsection, we focus on determining the optimal degree distribution based on our analysis in Section~\ref{sec:fullrank}  and~\ref{sec:recovery}. Formally, we can formulate the following optimization problem.
\begin{align}
&\min\quad \sum\limits_{k=1}^{mn} kp_k\label{eq:optimization}\\
& \begin{array}{r@{\quad}l@{}l@{\quad}l}
s.t. & \mathbb{P}(M\text{ is full rank})> p_c,\notag\\
&\left[1-\frac{\Omega_w'(x)}{mn}\right]^{mn+c}\leq 1-x-c_0\sqrt{\frac{1-x}{mn}},\notag\\
& x\in\left[0,1-b/mn\right], [p_k]\in \Delta_{mn},\notag
\end{array}
\end{align}
The objective is to minimize the average degree, namely, to minimize the computation and communication overhead at each worker. The first constraint represents that the probability of full rank is at least $p_c$. Since it is difficult to obtain the exact form of such a probability, we can use the analysis in Section~\ref{sec:fullrank} to replace this condition by requiring the probability that the balanced bipartite graph $G(V_1,V_2,P)$ contains a perfect matching is larger than a given threshold, which can be calculated exactly.
The second inequality represents the decodability condition that when $K=mn+c+1$ results are received, $mn-b$ blocks are recovered through the peeling decoding process and $b$ blocks are recovered from the rooting step (\ref{eq:borrow}). This  condition is modified from (\ref{eq:decode}) by adding an additional term, which is useful in increasing the expected ripple size~\cite{shokrollahi2006raptor}.  By discretizing the interval $[0,1-b/mn]$ and requiring the above inequality to hold on the discretization points, we obtain a set of linear inequalities constraints. Details regarding the exact form of the above optimization model and solutions are provided in the supplementary material.

\begin{figure}[t]
	\vskip -0.1in
	\begin{center}
		\centerline{\includegraphics[width=3.2in]{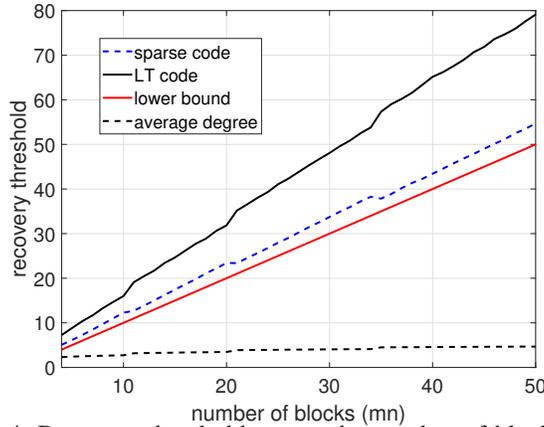}}
		\vskip -0.15in
		\caption{Recovery threshold versus the number of blocks $mn$.}
		\vskip -0.4in
		\label{fig:theorypractice}
	\end{center}
\end{figure}

\section{Experimental Results}

In this section, we present experimental results at Ohio Super Computing Center~\cite{OhioSupercomputerCenter1987}. We compare our proposed coding scheme against the following schemes: (i) \textbf{uncoded scheme}: the input matrices are divided uniformly across all workers and the master waits for all workers to send their results; (ii)  \textbf{sparse MDS code}~\cite{lee2017coded}: the generator matrix is a sparse random Bernoulli matrix with average computation overhead $\Theta(\ln(mn))$,  recovery threshold of $\Theta(mn)$ and decoding complexity $\tilde{O}(mn\cdot nnz(C))$. (iii) \textbf{product code}~\cite{lee2017high}: two-layer MDS code that can achieves the probabilistic recovery threshold of $\Theta(mn)$ and decoding complexity $\tilde{O}(rt)$. We use the above sparse MDS code to ensemble the product code to reduce the computation overhead. (iv) \textbf{polynomial code}~\cite{yu2017polynomial}: coded matrix multiplication scheme with optimum recovery threshold; (v) \textbf{LT code}~\cite{luby2002lt}: rateless code widely used in broadcast communication. It has low decoding complexity due to the peeling decoding algorithm. To simulate straggler effects in large-scale system, we randomly pick number of $s$ workers that are running a background thread which increases the computation time.

We implement all methods in python using MPI4py. To simplify the simulation, we fix the number of workers $N$ and randomly generate a coefficient matrix $M\in\mathbb{R}^{N\times mn}$ under given degree distribution offline such that it can resist one straggler. Then, each worker loads a certain number of partitions of input matrices according to the coefficient matrix $M$. In the computation stage, each worker computes the product of their assigned submatrices and returns the results using \texttt{Isend()}. Then the master node actively listens to the responses from each worker via \texttt{Irecv()}, and uses \texttt{Waitany()} to keep polling for the earliest finished tasks. Upon receiving enough results, the master stops listening and starts decoding the results.

\begin{figure}[htb]
	\vskip -0.15in
	\begin{center}
		\centerline{\includegraphics[width=5in]{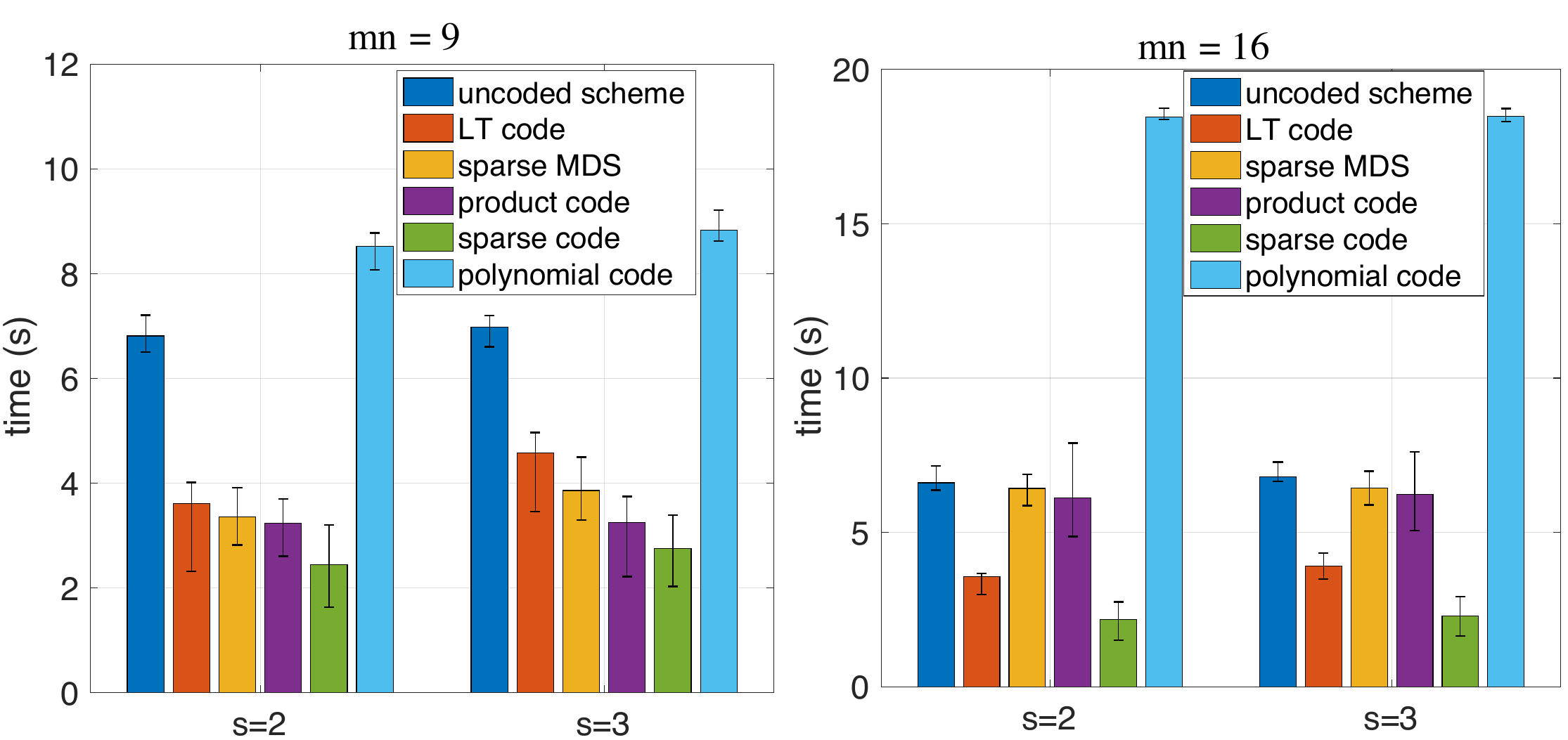}}
		\vskip -0.05in
		\caption{Job completion time for two $1.5$E$5\times 1.5$E$5$ matrices with $6$E$5$ nonzero elements. }
		\label{fig:simres_tot}
	\end{center}
	\vskip -0.40in
\end{figure}

\begin{figure*}[t]
	\vskip -0.1in
	\begin{center}
		\centerline{\includegraphics[width=5.1in]{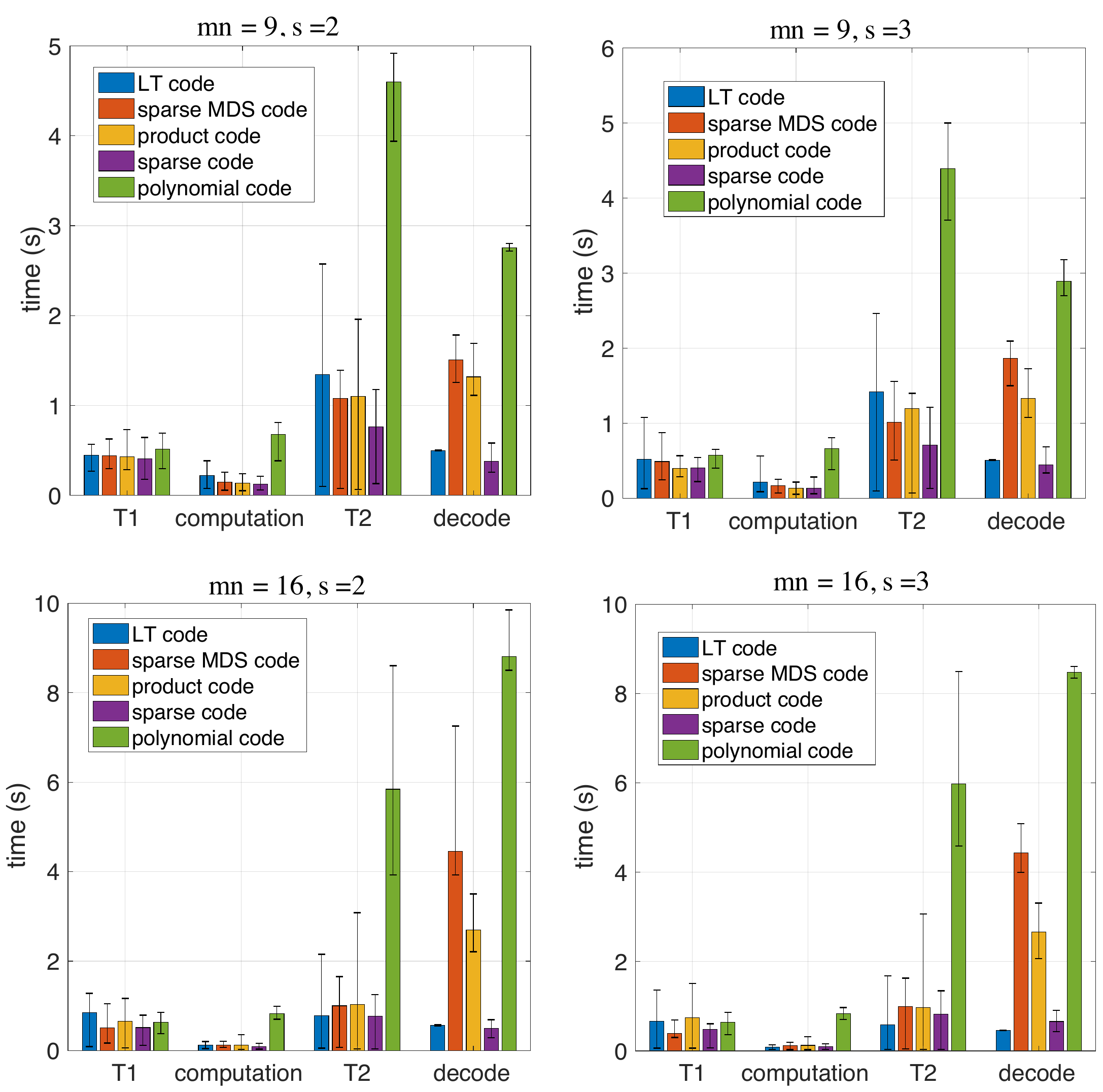}}
		\vskip -0.05in
		\caption{Simulation results for two $1.5$E$5\times 1.5$E$5$ matrices with $6$E$5$ nonzero elements. T1 and T2 are the transmission times from master to worker, and the worker to master, respectively. }
		\label{fig:simres_type}
	\end{center}
	\vskip -0.40in
\end{figure*}

We first generate two random Bernoulli sparse matrices with $r=s=t=150000$ and $600000$ nonzero elements. Figure.~\ref{fig:simres_tot} reports the job completion time under $m=n=3$, $m=n=4$ and number of stragglers $s=2,3$, based on $20$ experimental runs. It can be observed that our proposed sparse code requires the minimum time, and outperforms LT code (in 20-30\% the time), sparse MDS code and product code (in 30-50\% the time) and polynomial code (in 15-20\% the time). The uncoded scheme is faster than the polynomial code. The main reason is that, due to the increased number of nonzero elements of coded matrices, the per-worker computation time for these codes is increased. Moreover, the data transmission time is also greatly increased, which leads to additional I/O contention at the master node. 

We further compare our proposed sparse code with the existing schemes from the point of view of the time required to communicate inputs to each worker, compute the matrix multiplication in parallel, fetch the required outputs, and decode. As shown in Figure.~\ref{fig:simres_type}, in all of these component times, the sparse code outperforms the product code and polynomial code, with the effects being more pronounced for the transmission time and the decoding time. Moreover, due to the efficiency of proposed  hybrid decoding algorithm, our scheme achieves much less decoding time compared to the sparse MDS code and product code. Compared to the LT code, our scheme has lower transmission time because of the much lower recovery threshold. For example, when $m=n=4$ and $s=2$, the proposed sparse code requires $18$ workers, however, the LT code requires $24$ workers in average.

\begin{table*}[t]
	
	\caption{Data Statistics for Different Input Matrix }
	\label{tab:data}
	\begin{center}
		\begin{tabular}{|c|c|c|c|c|c|c|}
			\hline
			Data & $r$ & $s$ & $t$ & $nnz(A)$ &$nnz(B)$ & $nnz(C)$  \\
			\hline
			square & 1.5E5 & 1.5E5  & 1.5E5  & 6E5 & 6E5 & 2.4E6 \\
			\hline
			tall & 3E5 & 1.5E5  & 3E6 & 6E5 & 6E5  & 2.4E6  \\
			\hline
			fat &  1.5E5 & 3E5&  1.5E5  & 6E5  & 6E5 & 1.2E6  \\
			\hline
			amazon-08 / web-google &  735320 & 735323 & 916428  & 5158379 & 4101329 & 28679400 \\
			\hline
			cont1 /  cont11 & 1918396 & 1468599 & 1961392 &  2592597 & 5382995 & 10254724   \\
			\hline
			cit-patents / patents & 3774768 & 3774768 & 3774768 & 16518948 & 14970767 & 64796579  \\
			\hline
			hugetrace-00 / -01 & 4588484  &4588484  & 12057440 & 13758266 & 13763443 & 38255405  \\
			\hline
		\end{tabular}
		
	\end{center}
	\vskip -0.25in
\end{table*} 

\begin{table*}[t]
	\caption{Timing Results for Different Sparse Matrix Multiplications (in sec)}
	\label{tab:mainresult}
	\begin{center}
		\begin{tabular}{|c|c|c|c|c|c|c|c|c|c|}
			\hline
			Data &  uncoded & LT code & sparse MDS & product code & polynomial  & sparse code \\
			\hline
			square &  6.81 & 3.91  & 6.42  & 6.11 & 18.44 & \textbf{2.17} \\
			\hline
			tall & 7.03 & 2.69  & 6.25 & 5.50 & 18.51  & \textbf{2.04} \\
			\hline
			fat &  6.49 & 1.36 &  3.89  & 3.08  & 9.02 & \textbf{1.22} \\
			\hline
			amazon-08 / web-google &  15.35 & 17.59 & 46.26 & 38.61 & 161.6 & \textbf{11.01}\\
			\hline
			cont1 /  cont11 & 7.05 & 5.54 & 9.32 &  14.66 & 61.47 & \textbf{3.23} \\
			\hline
			cit-patents / patents & 22.10 & 29.15 & 69.86 & 56.59 & 1592 & \textbf{21.07} \\
			\hline
			hugetrace-00 / -01 & 18.06  & 21.76  & 51.15 & 37.36 & 951.3 & \textbf{14.16}  \\
			\hline
		\end{tabular}
		
	\end{center}
	\vskip -0.25in
\end{table*} 

We finally compare our scheme with these schemes for other type of matrices and larger matrices. The data statistics are listed in Table~\ref{tab:mainresult}. The first three data sets \texttt{square}, \texttt{tall} and \texttt{fat} are randomly generated square, fat and tall matrices. We also consider $8$ sparse matrices from real data sets~\cite{davis2011university}. We evenly divide each input matrices into $m=n=4$ submatrices and  number of stragglers is equal to $2$. We match the column dimension of $A$ and row dimension of $B$ using the smaller one. The timing results are results averaged over $20$ experimental runs. Among all experiments, we can observe in Table~\ref{tab:mainresult} that our proposed sparse code speeds up $1-3\times$ of uncoded scheme and outperforms the existing codes, with the effects being more pronounced for the real data sets. The job completion of the LT code, random sparse, product code is smaller than uncoded scheme in \texttt{square}, \texttt{tall} and \texttt{fat} matrix and  larger than uncoded scheme in those real data sets. 

\section{Conclusion}

In this paper, we proposed a new coded matrix multiplication scheme, which achieves near optimal recovery threshold, low computation overhead, and decoding time linear in the number of nonzero elements. Both theoretical and simulation results exhibit order-wise improvement of the proposed sparse code compared with the existing schemes. In the future, we will extend this idea to the case of higher-dimensional linear operations such as tensor operations.

\bibliographystyle{IEEEtran}
\bibliography{example_paper}

\appendix

\subsection{Proof of Theorem~\ref{thm:fullrank}}

Before we presenting the main proof idea, we first analyze the moment characteristic of our proposed Wave Soliton Distribution. For simplicity, we use $d$ to denote $mn$ in the sequel.
\begin{lemma}\label{lm:moment}
	Let a random variable $X$ follows the following Wave Soliton distribution $P_w=[p_1,p_2,\ldots,p_d]$.
	\begin{equation}
	p_k=\left\{
	\begin{aligned}
	&\frac{\tau}{d}, k=1\\
	&\frac{\tau}{70}, k=2\\
	&\frac{\tau}{k(k-1)}, 3\leq k\leq d
	\end{aligned}
	\right..
	\end{equation}
	Then all orders moment of $X$ is given by
	\begin{equation}
	\mathbb{E}[X^s]=\left\{
	\begin{aligned}
	&\Theta\left(\tau\ln\left(d\right)\right), s=1 \\
	&\Theta\left(\frac{\tau}{s}d^{s-1}\right), s\geq 2.
	\end{aligned}
	\right..
	\end{equation}
	\begin{proof}
		Based on the definition of moment for discrete random variable, we have
		\begin{align}
		\mathbb{E}[X]&=\sum\limits_{k=1}^dkp_k=\frac{\tau}{d}+\frac{\tau}{35}+\sum\limits_{k=3}^{d}\frac{\tau}{k-1}=\Theta\left(\tau\ln\left(d\right)\right).
		\end{align}
		Note that in the last step, we use the fact that $1+1/2+\cdots+1/d=\Theta(\ln(d))$.
		\begin{align}
		\mathbb{E}[X^s]&=\sum\limits_{k=1}^dk^sp_k=\frac{\tau}{d}+\frac{\tau 2^s}{70}+\sum\limits_{k=3}^{d}\frac{\tau k^{s-1}}{k-1}\notag\\
		&\overset{(a)}{=}\Theta\left(\frac{\tau}{s}d^{s-1}\right).
		\end{align}
		The step (a) uses the Faulhaber's formula that $\sum_{k=1}^{d}k^s=\Theta(d^{s+1}/(s+1))$ .
	\end{proof}
\end{lemma} 

The technical idea in the proof of this theorem is to use the Hall's theorem. Assume that the bipartite graph $G(V_1,V_2, P_w)$ does not have a perfect matching. Then by Hall's condition, there exists a violating set $S\subseteq V_1$ or  $S\subseteq V_2$ such that $|N(S)|<|S|$, where the neighboring set $N(S)$ is defined as $N(S)=\{y|(x,y)\in E(G) \text{ for some } x\in S\}$. Formally, by choosing such $S$ of smallest cardinality, one immediate consequence is the following technical statement.
\begin{lemma}\label{lm:halltheorem}
	If the bipartite graph $G(V_1,V_2, P_w)$ does not contain a perfect matching and $|V_1|=|V_1|=d$, then there exists a set $S\subseteq V_1$ or $S\subseteq V_2$ with the following properties.
	\begin{enumerate}
		\item $|S|=|N(S)|+1$.
		\item For each vertex $t\in N(S)$, there exists at least two adjacent vertices in $S$.
		\item $|S|\leq d/2$.
	\end{enumerate}
\end{lemma} 
Figure~\ref{fig:example} illustrates two simple examples of structure $S$ satisfying above three conditions. 

\textbf{Case 1:} We consider that $S\subseteq V_1$. Define an event $E(V_1)$ is that there exists a set $S\subseteq V_1$ satisfying above three conditions. 

\textbf{Case 1.1:} We consider $S\subseteq V_1$ and $|S|=1$. 

In this case, we have $|N(S)|=0$ and need to estimate the probability that there exists one isolated vertex in partition $V_1$.  Let random variable $X_i$ be the indicator function of the event that vertex $v_i$ is isolated. Then we have the probability that
\begin{equation*}
\mathbb{P}(X_i=1)=\left(1-\frac{\alpha}{d}\right)^d,
\end{equation*}
where $\alpha$ is the average degree of a node in the partition $V_2$ and $\alpha=\Theta\left(\tau\ln\left(d\right)\right)$ from Lemma~\ref{lm:moment}. Let $X$ be the total number of isolated vertices in partition $V_1$. Then we have
\begin{align}
\mathbb{E}[X] = {E}\left[\sum\limits_{i=1}^{d}X_i\right]=d\left(1-\frac{\alpha}{d}\right)^d=d\left(1-\frac{\tau\ln(d)}{d}\right)^d=\Theta\left(\frac{1}{d^{\tau-1}}\right)\overset{(a)}{=}o(1).
\end{align}
The above, step (a) is based on the fact that $\tau>1.94$.

\begin{figure*}[t]
	\vskip 0.1in
	\begin{center}
		\centerline{\includegraphics[width=5in]{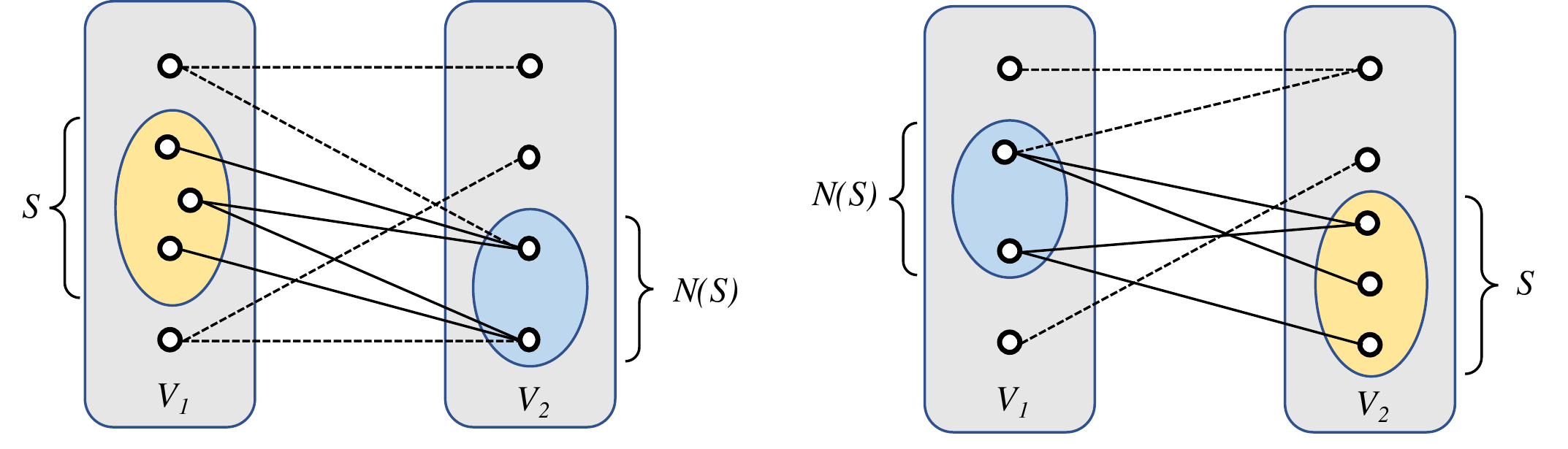}}
		\vskip -0.0in
		\caption{Example of structure $S\in V_1$, $N(S)\in V_2$ and $S\in V_2$, $N(S)\in V_1$ satisfying condition 1,2 and 3. One can easily check that there exists no perfect matching in these two examples.}
		\label{fig:example}
	\end{center}
	\vskip -0.3in
\end{figure*}

Before we presenting the results in the case $2 \leq |S| \leq d/2$, we first define the following three events.
\begin{definition}
	Given a set $S\subseteq V_1$ and $|S|=s$, for each vertex $v\in V_2$, define an event $S^s_{0}$ is that $v$ has zero adjacent vertex in $S$, an event $S^s_{1}$ is that $v$ has one adjacent vertex in $S$ and an event $S^s_{\geq2}$ is that $v$ has at least two adjacent vertices in $S$. 
\end{definition}
Then we can upper bound the probability of event $E$ by
\begin{align}\label{eq:mainupperbound1}
\mathbb{P}(E(V_1))&=\mathbb{P}(\text{there exists $S\in V_1$ and $N(S)\in V_2$ such that conditions $1,2$ and $3$ are satisfied})\notag\\
&\overset{(a)}{\leq} \mathbb{P}(\text{there exists an isoloted vertex in $V_1$})+ \sum\limits_{s=2}^{d/2}\binom{d}{s}\binom{d}{s-1}\cdot\mathbb{P}(S^s_{\geq 2})^{s-1}\cdot\mathbb{P}(S^s_{0})^{d-s+1}\notag\\
&=o(1)+ \sum\limits_{s=2}^{d/2}\binom{d}{s}\binom{d}{s-1}\cdot\mathbb{P}(S^s_{\geq 2})^{s-1}\cdot\mathbb{P}(S^s_{0})^{d-s+1}.
\end{align}
The above, step (a) is based on the union bound. Formally, given $|S|=s$ and fixed vertex $v\in V_2$, we can calculate $\mathbb{P}(S^s_{0})$ via the law of total probability.
\begin{align}\label{eq:notconnected}
\mathbb{P}(S^s_{0})&=\sum\limits_{k=1}^{d}\mathbb{P}(S^s_{0}|\deg(v)=k)\cdot \mathbb{P}(\deg(v)=k)\notag\\
&=\sum\limits_{k=1}^{d-s}p_k\binom{d-s}{k}\cdot\binom{d}{k}^{-1}\notag\\
&=\sum\limits_{k=1}^{d-s}p_k\left(1-\frac{k}{d}\right)\left(1-\frac{k}{d-1}\right)\left(1-\frac{k}{d-2}\right)\cdots\left(1-\frac{k}{d-s+1}\right).
\end{align}

Similarly, the probability $\mathbb{P}(S^s_{1})$  is given by the following formula.
\begin{align}\label{eq:degree1ins}
\mathbb{P}(S^s_{1})&=\sum\limits_{k=1}^{d}\mathbb{P}(S^s_{1}|\deg(v)=k)\cdot \mathbb{P}(\deg(v)=k)\notag\\
&=\sum\limits_{k=1}^{d-s+1}p_k\binom{d-s}{k-1}\cdot\binom{s}{1}\cdot\binom{d}{k}^{-1}\notag\\
&=\sum\limits_{k=1}^{d-s+1}p_k\left(1-\frac{k}{d}\right)\left(1-\frac{k}{d-1}\right)\cdots\left(1-\frac{k}{d-s+2}\right)\cdot\frac{sk}{d-s+1}.
\end{align}

Then, the probability $\mathbb{P}(S^s_{\geq 2})$ is given by the following formula.
\begin{equation}\label{eq:atleast2}
\mathbb{P}(S^s_{\geq 2})=1-\mathbb{P}(S^s_{0})-\mathbb{P}(S^s_{1}).
\end{equation}
The rest is to utilize the formula (\ref{eq:notconnected}), (\ref{eq:degree1ins}) and (\ref{eq:atleast2}) to estimate the order of (\ref{eq:mainupperbound1}) under several scenarios.

\textbf{Case 1.2:} We consider $S\subseteq V_1$ and $|S|=\Theta(1)$. 

Based on the result of (\ref{eq:notconnected}), we have
\begin{align}
\mathbb{P}(S^s_{0})&\leq\sum\limits_{k=1}^{d}p_k\left(1-\frac{k}{d}\right)^s=\sum\limits_{k=1}^{d}p_k\sum\limits_{i=0}^{s}\binom{s}{i}(-1)^i\frac{k^i}{d^i}\notag\\
&\overset{(a)}{=}\sum\limits_{k=1}^{d}p_k\left(1-\frac{sk}{d}\right)+\sum\limits_{i=2}^{s}\binom{s}{i}\frac{(-1)^i}{d^i}\sum\limits_{k=1}^{d}k^ip_k\notag\\
&\overset{(b)}{=}1-\Theta\left(\frac{s\tau\ln(d)}{d}\right)+\Theta\left(\frac{1}{d}\right)=1-\Theta\left(\frac{s\tau\ln(d)}{d}\right).\label{eq:upproba0}
\end{align}
The above, step (a) is based on exchanging the order of summation; step (b) utilizes the result of Lemma~\ref{lm:moment} and the fact that $s$ is the constant. Similarly, we have
\begin{align}
\mathbb{P}(S^s_{\geq 2})&\leq 1-\mathbb{P}(S^s_{0})\leq 1-\sum\limits_{k=1}^{d-s}p_k\left(1-\frac{k}{d-s}\right)^s=\Theta\left(\frac{s\tau\ln(d)}{d}\right).\label{eq:upproba1}
\end{align}
Combining the upper bound (\ref{eq:mainupperbound1}) and estimation of $\mathbb{P}(S^s_{0})$ and $\mathbb{P}(S^s_{\geq 2})$, we have
\begin{align}
&\sum\limits_{s=\Theta(1)}\binom{d}{s}\binom{d}{s-1}\cdot\mathbb{P}(S^s_{\geq 2})^{s-1}\cdot\mathbb{P}(S^s_{0})^{n-s+1}\notag\\
\overset{(a)}{\leq}&\sum\limits_{s=\Theta(1)}\frac{e^{2s-1}d^{2s-1}}{s^{2s-1}}\cdot\Theta\left(\frac{s\tau\ln(d)}{d}\right)^{s-1}\cdot\left[1-\Theta\left(\frac{s\tau\ln(d)}{d}\right)\right]^{d-s+1}\notag\\
=&\sum\limits_{s=\Theta(1)}\Theta\left(\frac{\ln^{s-1}(d)}{d^{(\tau-1)s}}\right)=o(1).\label{eq:estimate1.2}
\end{align}
The above, step (a) utilizes the inequality $\binom{d}{s}\leq(ed/s)^s$.

\textbf{Case 1.3:} We consider $S\subseteq V_1$, $|S|=\Omega(1)$ and $|S|=o(d)$ . 

Based on the result of (\ref{eq:notconnected}), we have
\begin{align}
\mathbb{P}(S^s_{0})&\leq\sum\limits_{k=1}^{d}p_k\left(1-\frac{k}{d}\right)^s\notag\\
&\overset{(a)}{=}\sum\limits_{ks=o(d)}p_k\left(1-\frac{ks}{d}\right)+\sum\limits_{ks=\Theta(d) \text{ or } \Omega(d)}p_k\left(1-\frac{k}{d}\right)^s\notag\\
&\leq\sum\limits_{k=1}^{d/s}p_k\left(1-\frac{ks}{d}\right)+\sum\limits_{ks=\Theta(d) \text{ or } \Omega(d)}p_k\notag\\
&\overset{(b)}{\leq} 1-\frac{\tau (s-1)}{d}-\Theta\left(\frac{\tau s\ln(d/s)}{d}\right)+\Theta\left(\frac{\tau s}{d}\right)\notag\\
&\overset{(c)}{=}1-\Theta\left(\frac{\tau s\ln(d/cs)}{d}\right).
\end{align}
The above, step (a) is based on summation over different orders of $k$. In particular, when $ks=o(d)$ and $s=\Omega(1)$, we have $(1-k/d)^s=$ $\Theta(e^{-ks/d})=1-\Theta(ks/d)$. Step (b) utilizes the partial sum formula $1+1/2+\cdots+s/d=\Theta(\ln(d/s))$. The parameter $c$ of step (c) is a constant. Similarly, we have
\begin{align}
\mathbb{P}(S^s_{\geq 2})&\leq 1-\mathbb{P}(S^s_{0})\notag\\
&\overset{(a)}{\leq} 1-\left[\sum\limits_{ks=o(d)}p_k\left(1-\frac{ks}{d-s}\right)+\sum\limits_{ks=\Theta(d),k\leq d/s-1}p_k e^{-\frac{ks}{d-s}}\right]\notag\\
&\overset{(b)}{\leq}1-\sum\limits_{k=1}^{d/s-1}p_k\left(1-\frac{ks}{d-s}\right)\notag\\
&\overset{(c)}{=}\Theta\left(\frac{\tau s\ln(d/c's)}{d}\right).
\end{align}
The above, step (a) is based on summation over different orders of $k$, and abandon the terms when $k\geq d/s$. In particular, when $ks=\Theta(d)$ and $s=\Omega(1)$, we have $(1-k/(d-s))^s=$ $\Theta(e^{-ks/(d-s)})$. The step (b) utilizes the inequality $e^{-x}\geq 1-x, \forall x\geq 0$. The parameter $c'$ of step (c) is a constant.  Combining the upper bound (\ref{eq:mainupperbound1}) and estimation of $\mathbb{P}(S^s_{0})$ and $\mathbb{P}(S^s_{\geq 2})$, we have
\begin{align}
&\sum\limits_{s=\Omega(1),s=o(d)}\binom{d}{s}\binom{d}{s-1}\cdot\mathbb{P}(S^s_{\geq 2})^{s-1}\cdot\mathbb{P}(S^s_{0})^{d-s+1}\notag\\
\leq&\sum\limits_{s=\Omega(1),s=o(d)}\frac{e^{2s-1}d^{2s-1}}{s^{2s-1}}\cdot\Theta\left(\frac{s\tau\ln(d/c's)}{d}\right)^{s-1}\cdot\left[1-\Theta\left(\frac{s\tau\ln(d/cs)}{d}\right)\right]^{d-s+1}\notag\\
=&\sum\limits_{s=\Omega(1),s=o(d)} \Theta\left(\left(\frac{s}{d}\right)^{(\tau-1)s} c^{\tau s}\tau^{s-1}e^{2s-1}\ln^{s-1}(d/c's)\right)\notag\\
=&\sum\limits_{s=\Omega(1),s=o(d)} \Theta\left(\frac{s\ln^{1.06}(d/s)}{d}\right)^{(\tau-1)s}=o(1).
\label{eq:estimate1.3}
\end{align}

\textbf{Case 1.4:} We consider $S\subseteq V_1$ and $|S|=\Theta(d)=cd$ . 

Based on the result of (\ref{eq:notconnected}) and Stirling's approximation, we have
\begin{align}\label{eq:approxnotconnected}
\mathbb{P}(S^s_{0})&=\sum\limits_{k=1}^{d-s}p_k\left(1-\frac{k}{d}\right)^{d-k+\frac{1}{2}}\left(1+\frac{k}{d-s-k}\right)^{d-s-k+\frac{1}{2}}(1-c)^k\notag\\
&\overset{(a)}{=}\sum\limits_{k=o(d)}p_k(1-c)^k+\sum\limits_{k=\Theta(d),k\leq d-s}o\left((1-c)^k\right)\notag\\
&=p_1(1-c)+p_2(1-c)^2+\tau\sum\limits_{k\geq 3,k=o(d)}\frac{(1-c)^k}{k(k-1)}+o(1) \notag\\
&\overset{(b)}{=}p_1(1-c)+p_2(1-c)^2+\tau\left[\frac{1}{2}(1-c^2)+c\ln(c)\right]\triangleq f_0(c).
\end{align}
The above, step (a) is based on summation over different orders of $k$. The step (b) is based on the following partial sum formula.
\begin{align}\label{eq:partialsum}
\sum\limits_{k=3}^{q}\frac{(1-c)^k}{k(k-1)}&=\frac{1}{2q}\left[2c(1-c)q(1-c)^q\Phi(1-c,1,q+1)-2(1-c)^{q+1}+q(1-c^2)+2cq\ln(c)\right].
\end{align}
where the function $\Phi(1-c,1,q+1)$ is the Lerch Transcendent, defined as
\begin{equation}
\Phi(1-c,1,q+1)=\sum\limits_{k=0}^{\infty}\frac{c^k}{k+q+1}.
\end{equation}
Let $q=\Omega(1)$, we arrive at the step (b). Similarly, utilizing the result of (\ref{eq:degree1ins}), we have
\begin{align}
\mathbb{P}(S^s_{1})&=\frac{c}{1-c}\sum\limits_{k=o(d)}p_kk(1-c)^k\notag\\
&=p_1c+2p_2c(1-c)+\tau c(c-1-\ln(c))\triangleq f_1(c).
\end{align}

Therefore, utilizing the upper bound (\ref{eq:mainupperbound1}), we arrive at
\begin{align}
&\sum\limits_{s=\Theta(d),s\leq d/2}\binom{d}{s}\binom{d}{s-1}\cdot\mathbb{P}(S^s_{\geq 2})^{s-1}\cdot\mathbb{P}(S^s_{0})^{n-s+1}\notag\\
\leq&\sum\limits_{s=cd,s\leq d/2}\left[\left(\frac{1}{c}\right)^{2c}\left(\frac{1}{1-c}\right)^{2(1-c)}\left[1-f_0(c)-f_1(c)\right]^{c}\left[f_0(c)\right]^{1-c}\right]^d\notag\\
=&\sum\limits_{s=cd,s\leq d/2}(1-\Theta(1))^d=o(1).\label{eq:estimate1.4}
\end{align}
Therefore, combining the results in the above four cases, we conclude that $\mathbb{P}(E(V_1))=o(1)$.

\textbf{Case 2:} We consider that $S\subseteq V_2$. We relax the condition 2 in Lemma~\ref{lm:halltheorem} to the following condition.

$\quad 2'.$ \emph{For each vertex $t\in S$, there exists at least one adjacent vertex in $N(S)$.}

Define an event $E(V_2)$ is that there exists a set $S\subseteq V_2$ satisfying condition 1, 2, 3, and an event $E'$ is that there exists a set $S$ satisfying above condition $1, 2'$ and $3$.  One can easily show that the event $E(V_2)$ implies the event $E'$ and $\mathbb{P}(E(V_2))\leq \mathbb{P}(E')$. Then we aim to show that the probability of event $E'$ is $o(1)$.
\begin{definition}
	Given a set $S\subseteq V_2$ and $|S|=s$, for each vertex $v\in V_2$, define an event $N^s_{\geq1}$ is that $v$ has at least one adjacent vertex in $N(S)$ and $v$ does not connect to any vertices in $V_1/N(S)$.
\end{definition}
Then we can upper bound the probability of event $E'$ by
\begin{align}\label{eq:mainupperbound2}
\mathbb{P}(E')&=\mathbb{P}(\text{there exists $S\in V_2$ and $N(S)\in V_1$ such that condition $1,2'$ and $3$ are satisfied})\notag\\
&\overset{(a)}{\leq}\sum\limits_{s=2}^{d/2}\binom{d}{s}\binom{d}{s-1}\cdot\mathbb{P}(N^s_{\geq 1})^{s}\notag\\
&\leq \frac{e^{2s-1}d^{2s-1}}{s^{2s-1}}\cdot\mathbb{P}(N^s_{\geq 1})^{s}.
\end{align}
The above, step (a) is based on the fact that any vertices in set $V_2$ has degree at least one according to the definition of the Wave Soliton distribution. Given $|S|=s$ and fixed vertex $v\in S$, we can calculate $\mathbb{P}(N^s_{\geq 1})$ via the law of total probability.
\begin{align}\label{eq:nslarge2}
\mathbb{P}(N^s_{\geq 1})&=\sum\limits_{k=1}^{d}\mathbb{P}(N^s_{\geq 1}|\deg(v)=k)\cdot \mathbb{P}(\deg(v)=k)\notag\\
&=\sum\limits_{k=1}^{s-1}p_k\binom{s-1}{k}\cdot\binom{d}{k}^{-1}\notag\\
&=\sum\limits_{k=1}^{s-1}p_k\frac{s-1}{d}\cdot\frac{s-2}{d-1}\cdot\frac{s-3}{d-2}\cdots\cdot\frac{s-k}{d-k+1}\notag\\
&\leq \sum\limits_{k=1}^{s-1}p_k \frac{s^k}{d^k}.
\end{align}

\textbf{Case 2.1:} We consider $S\subseteq V_2$ and $|S|=\Theta(1)$.

Based on the result of (\ref{eq:nslarge2}), we have
\begin{align}
\mathbb{P}(N^s_{\geq 1})&\leq \frac{\tau s}{d^2} + p_2 \frac{s^2}{d^2}+\sum\limits_{k=3}^{s-1}p_k \frac{s^k}{d^k}\leq  \frac{\tau s}{d^2} +p_2\frac{s^2}{d^2}+\frac{\tau s^3}{d^3}\left(\frac{1}{2}-\frac{1}{s-1}\right)<\frac{s^2}{d^2}\left[\frac{1}{36}+\frac{\tau}{s}+\frac{\tau s}{d}\left(\frac{1}{2}-\frac{1}{s-1}\right)\right]
\end{align}
Then we have
\begin{align}
\sum\limits_{s=\Theta(1)}\frac{e^{2s-1}d^{2s-1}}{s^{2s-1}}\cdot\mathbb{P}(N^s_{\geq 1})^{s}=\Theta\left(\frac{1}{d}\right).
\end{align}

\textbf{Case 2.2:} We consider $S\subseteq V_2$, $|S|=\Omega(1)$ and $|S|=o(d)$.

Similarly, using the result in Case 2.1 and  upper bound (\ref{eq:mainupperbound2}), we arrive
\begin{align}
&\sum\limits_{s=o(d)}\binom{d}{s}\binom{d}{s-1}\cdot\mathbb{P}(N^s_{\geq 1})^{s}\leq\sum\limits_{s=o(d)}\frac{se^{2s-1}}{d}\left[\frac{1}{36}+\frac{\tau}{s}+\frac{\tau s}{d}\left(\frac{1}{2}-\frac{1}{s-1}\right)\right]^s\overset{(a)}{=}o(1).
\end{align}
The above, step (a) is based on the fact that $1/36+o(1)<e^{-2}$.

\textbf{Case 2.3:} We consider $S\subseteq V_2$ and $|S|=\Theta(d)=cd$.
Based on the result of (\ref{eq:nslarge2}), we have
\begin{align}
\mathbb{P}(N^s_{\geq 1})&\leq\sum\limits_{k=1}^{s-1}p_k c^k\overset{(a)}{\leq}p_1c+p_2c^2+\tau\left(c-\frac{c^2}{2}+(1-c)\ln(1-c)\right)\triangleq f_2(c).
\end{align}
The above, step (a) utilizes the partial sum formula (\ref{eq:partialsum}). Using the upper bound (\ref{eq:mainupperbound2}), we arrive
\begin{align}
&\sum\limits_{s=\Theta(d)}\binom{d}{s}\binom{d}{s-1}\cdot\mathbb{P}(N^s_{\geq 2})^{s}=\left[\left(\frac{1}{c}\right)^{2c}\left(\frac{1}{1-c}\right)^{2(1-c)}[f_2(c)]^c\right]^d=\sum\limits_{s=\Theta(d)}(1-\Theta(1))^d=o(1).
\end{align}

Combining the results in the above three cases, we have $\mathbb{P}(E')=o(1)$. Therefore, the theorem follows.

\subsection{Proof of Lemma~\ref{lm:decode}}

Consider a random bipartite graph generated by degree distribution $P_w$ of the nodes in left partition $V_2$, define a left edge degree distribution $\lambda(x)=\sum_{k}\lambda_kx^{k-1}$ and a right edge degree distribution $\rho(x)=\sum_{k}\rho_kx^{k-1}$, where $\lambda_k$ ($\rho_k$) is the fraction of edges adjacent to a node of degree $k$ in the left partition $V_1$ (right partition $V_2$).  The existing analysis in~\cite{luby2001efficient} provides a quantitative condition regarding the recovery threshold in terms of $\lambda(x)$ and $\rho(x)$.
\begin{lemma}\label{lm:luby}
	Let a random bipartite graph be chosen at random with left edge degree distribution $\lambda(x)$ and right edge degree distribution $\rho(x)$, if
	\begin{equation}
	\lambda(1-\rho(1-x)) < x, x\in[\delta,1],
	\end{equation}
	then the probability that peeling decoding process cannot recover $\delta d$ or more of original blocks is upper bounded by $e^{-cd}$ for some constant $c$.
\end{lemma}

We first derive the edge degree distributions $\lambda(x)=\sum_{k}\lambda_kx^{k-1}$ and $\rho(x)=\sum_{k}\rho_kx^{k-1}$ via the degree distribution $\Omega_w(x)$. Suppose that the recovery threshold is $K$. The total number of edges is $K\Omega_w'(1)$ and the total number of edges that is adjacent to a right node of degree $k$ is $Kkp_k$. Then, we have $\rho_k=kp_k/\Omega_w'(1)$ and 
\begin{equation}
\rho(x)=\Omega_w'(x)/\Omega_w'(1).
\end{equation}

Fix a node $v_i\in V_1$, the probability that node $v_i$ is a neighbor of node $v_j\in V_2$ is given by 
\begin{equation*}
\sum\limits_{k=1}^{d}p_k\binom{d-1}{k-1}\binom{d}{k}^{-1}=\frac{1}{d}\sum\limits_{k=1}^{d}kp_k=\frac{\Omega_w'(1)}{d}.
\end{equation*}
Since $|V_2|=K$, the probability that node $v_i$ is the neighbor of exactly $l$ nodes in $V_2$ is $\binom{K}{l}(\Omega_w'(1)/d)^l(1-\Omega_w'(1)/d)^{K-l}$ and corresponding probability generating function is
\begin{equation*}
\sum\limits_{l=1}^K\binom{K}{l}\left(\frac{\Omega_w'(1)}{d}\right)^l\left(1-\frac{\Omega_w'(1)}{d}\right)^{K-l}x^l=\left[1-\frac{\Omega_w'(1)(1-x)}{d}\right]^K.
\end{equation*}
Then we can obtain $\lambda(x)$ as
\begin{equation}
\lambda(x)=\left[1-\frac{\Omega_w'(1)(1-x)}{d}\right]^{K-1}.
\end{equation}
Further, we have
\begin{equation}
\lambda(1-\rho(1-x))=\left[1-\frac{\Omega_w'(1-x)}{d}\right]^{K-1}.
\end{equation}
Combining these results with Lemma~\ref{lm:luby}, let $\delta=b/mn$, the lemma follows.

\subsection{Proof of Theorem~\ref{thm:decode}}

Suppose that $K=cd+1$, one basic fact is that
\begin{equation}
\lambda(1-\rho(1-x))=\left[1-\frac{\Omega_w'(1-x)}{d}\right]^{K-1}\leq e^{-c\Omega_w'(1-x)}
\end{equation}
Based on the results of Lemma 3, the rest is to show that $e^{-c\Omega_w'(x)}\leq 1-x$ for $x\in[0,1-b/d]$. Based on the definition of our Wave Soliton distribution, we have
\begin{align}
\Omega_w'(x)&=\frac{\tau}{d}+\frac{\tau x}{35}+\tau\sum\limits_{k=2}^{d-1}\frac{x^{k}}{k}\notag\\
&=\frac{\tau}{d}-\frac{ 34\tau x}{35}-\tau\ln(1-x)-\tau\sum\limits_{k=d}^{\infty}\frac{x^{k}}{k}\overset{(a)}{\geq}\frac{\tau}{d}-\frac{ 34\tau x}{35}-\tau\ln(1-x)-\tau x^{10}.
\end{align}
The above step (a) is utilizing the fact that $x^{10}\geq \sum_{k=d}^{\infty}\frac{x^{k}}{k}$ for $x\in[0,1-b/d]$. It remains to show that there exists a constant $c$ such that 
\begin{equation}
-c\left[\frac{\tau}{d}-\frac{ 34\tau x}{35}-\tau\ln(1-x)-\tau x^{10}\right]\leq \ln(1-x), \text{ for }  x\in[0,1-b/d].
\end{equation}
which is verified easily.

\subsection{Optimal Design of Sparse Code}

We focus on determining the optimal degree distribution based on our previous analysis. Formally, we can formulate the following optimization problem.
\begin{align}
&\min\quad \sum\limits_{k=1}^{mn} kp_k\label{eq:optimization}\\
& \begin{array}{r@{\quad}l@{}l@{\quad}l}
s.t. & \mathbb{P}(M\text{ is full rank})> p_c,\notag\\
&\left[1-\frac{\Omega_w'(x)}{mn}\right]^{mn+c}\leq 1-x-c_0\sqrt{\frac{1-x}{mn}}, x\in\left[0,1-b/mn\right], \notag\\
&[p_k]\in \Delta_{mn},\notag
\end{array}
\end{align}

Here the coefficient matrix $M$ has row dimension $mn$ and column dimension $mn+c$. Due to the hardness of estimating the probability that $M$ is full rank, we relax this condition to the following condition
\begin{equation}
\mathbb{P}(G(V_1,V_2,P)\text{ contains a perfect matching})> p_m,
\end{equation}
where $|V_1|=|V_2|=mn$ and $p_m$ is a given threshold. The basic intuition behind this relaxation comes from the analysis of Section 4.1: the probability that there exists a perfect matching in the balanced random bipartite graph $G(V_1,V_2,P)$ provides a lower bound of the probability that the coefficient matrix $M$ is of full rank. Based on this relaxation, the rest is to estimate the probability that $G(V_1,V_2,P)$ contains a perfect matching. Instead of estimating the lower bound of such a probability as in the proof of Theorem 2, here we provide an exact formula that is a function of the degree distribution $P$. 

In the sequel, we denote $mn$ by $d$. Suppose the vertex in partition $V_2$ is denoted by $\{v_1,v_2,\ldots,v_d\}$. Define a degree distribution $P^{(s)}=[p_0^{(s)}, p_1^{(s)}, p_2^{(s)}, \ldots, p_d^{(s)}]$, where $p_k^{(s)}$ is the probability that any vertex $v\in V_2$ has exact $k$ adjacent vertices in a given set $S\subseteq V_1$ with $|S|=s$. For example, we have $P^{(d)}=P$, where $P$ is the original degree distribution. Let $E(d,V_1,P)$ be the event that $G(V_1,V_2,P)$ contains a perfect matching and $|V_1|=d$. In order to calculate $\mathbb{P}(E(d,V_1,P))$, we condition this probability on that there exists $v_1^n\in V_1$  matching with  $v_1\in V_2$. Then, we have
\begin{align}
\mathbb{P}(E(d,V_1,P))=&\mathbb{P}(E(d,V_1,P)|\exists v_1^n\in V_1 \text{ matching with } v_1)\mathbb \cdot \mathbb{P}(\exists v_1^n\in V_1 \text{ matching with } v_1)\notag\\
=&\mathbb{P}(E(d-1,V_1\backslash \{v_1^n\}, P^{(d-1)}))\mathbb \cdot \left(1 - p^{(d)}_0\right)\notag\\
=&\left(1 - p^{(d)}_0\right)\cdot \mathbb{P}(E(d-1,V_1\backslash \{v_1^n\},P^{(d-1)})|\exists v_2^n\in V_1\backslash \{v_1^n\} \text{ matching with } v_2)\notag\\
&\qquad\qquad\qquad\qquad\qquad\qquad\qquad\qquad\qquad\qquad\mathbb{P}(\exists v_2^n\in V_1\backslash \{v_1^n\} \text{ matching with } v_2)\notag\\
=&\mathbb{P}(E(d-2,V_1\backslash \{v_1^n, v_2^n\}, P^{(d-2)}))\left(1 - p^{(d-1)}_0\right) \cdot \left(1 - p^{(d)}_0\right)\notag\\
=&\cdots\notag\\
=&\left(1 - p^{(1)}_0\right)\left(1 - p^{(2)}_0\right)\cdots\left(1 - p^{(d-1)}_0\right) \left(1 - p^{(d)}_0\right).
\end{align}
The rest is to estimate the degree evolution $P^{(s)}$. Similarly, we have the following recursive formula to calculate the degree evolution: for all $ 0\leq k\leq s, 1\leq s\leq d-1.$
\begin{align}
p^{(s)}_k&= p^{(s+1)}_k\binom{s}{k}\binom{s+1}{k}^{-1} + p^{(s+1)}_{k+1}\binom{s}{k}\binom{s+1}{k+1}^{-1}\notag\\
&=p^{(s+1)}_k\left(1-\frac{k}{s+1}\right) + p^{(s+1)}_{k+1}\frac{k+1}{s+1}.
\end{align}

Utilizing the fact that $P^{(d)}=P$, we can get the exact formula of $P^{(s)}$, $1\leq s\leq d-1$, and the formula of the probability that $G(V_1,V_2,P)$ contains a
 perfect matching.
 
 \begin{table*}[t]
 	\caption{Optimizaed Degree Distribution for Various $mn$ (Numbers in brackets are Results from Robust Soliton Distribution.)}
 	\label{tab:degree}
 	\begin{center}
 		\begin{tabular}{|c|c|c|c|c|c|c|c|c|c|c|}
 			\hline
 			$mn$ & $p_1$ & $p_2$  & $p_3$ & $p_4$ &  $p_5$ &$p_6$  & recovery threshold  & average degree & rooting step \\
 			\hline
 			6 & 0.0217 & 0.9390& 0.0393 & 0.0 & 0.0 & 0.0 & 7.54 (7.61) & 2.01 (2.04) & 0.84 (0.47)\\
 			\hline
 			9 & 0.0291  & 0.7243 & 0.2466 & 0.0& 0.0 & 0.0 & 11.81 (12.15) & 2.21 (2.20) & 0.90 (0.56)\\
 			\hline
 			12 & 0.0598 & 0.1639 & 0.7056 & 0.0707 & 0.0 &  0.0&14.19 (14.47) & 2.78 (2.78) & 1.47 (1.03)\\
 			\hline
 			16 & 0.0264 & 0.3724 & 0.1960 & 0.4052 &  0.0& 0.0 &19.11 (19.83) & 2.98 (2.91) & 1.68 (1.08)\\
 			\hline
 			25 & 0.0221 & 0.4725& 0.1501 &  0.0 & 0.0 &  0.3553 &  28.71 (29.12) & 3.54 (3.55) & 2.35 (2.12)\\
 			\hline
 		\end{tabular}
 	\end{center}
 	\vskip -0.2in
 \end{table*} 

TABLE~\ref{tab:degree} shows several optimized degree distributions we have found using model (\ref{eq:optimization}) with specific choice of parameters $c,c_0, b$ and $p_c$. We also include the several performance results under above distribution and the Robust Soliton distribution (RSD). In traditional RSD, the degree $2$ always have the highest probability mass, i.e., $p_2\approx 0.5$. It is interesting to note that our optimized distribution has a different shape, which depends on $mn$ and choices of parameters. We can observe that, under the same average degree, the optimized distribution has a lower recovery threshold and larger number of rooting steps compared to the RSD. Another observation in solving the optimization problem is that, when the parameter $p_m$ is increased, the recovery threshold of proposed sparse code will be decreased, and the average degree and the number of rooting steps will be increased. 

\end{document}